\documentclass[english,a4paper,leqno,10pt]{article}

\usepackage{ams math,amssymb,amsfonts,amsthm,mathrsfs,MnSymbol}
\usepackage{graphicx,epsfig,color}

\usepackage{setspace}
\usepackage{latexsym}
\usepackage{epsfig}
\usepackage{yfonts}

\usepackage{enumitem}
\setitemize{wide}

\newtheorem{Definition}{Definition}
\newtheorem{Proposition}{Proposition}
\newtheorem{Theorem}{Theorem}

\newtheorem{Lemma}{Lemma}

\newtheorem{Aux-Lemma}{Aux-Lemma}

\newcommand{\vs}{\vspace{0.2cm}}
\newcommand{\n}{\noindent}
\newcommand{\be}{\begin{equation}}
\newcommand{\ee}{\end{equation}}
\newcommand{\ben}{\begin{equation*}}
\newcommand{\een}{\end{equation*}}

\usepackage{babel} 
\usepackage{blindtext} 

\usepackage{etoolbox}

\usepackage{fancyhdr}
\AtBeginDocument{\thispagestyle{plain}}
\fancypagestyle{plain}{
\fancyhead{}
\fancyfoot{}
\fancyhead[LE,RO]{}
\fancyhead[LO,RE]{}
\fancyfoot[R]{\thepage}

}

\newcommand{\dif}{{\rm d}}
\newcommand{\nor}{\mathfrak{n}}

\newcommand{\AR}{{\mathcal A}_{\mathbb{R}^{3}}}
\newcommand{\dist}{{\rm{dist}}}

\newcommand{\length}{{\rm{length}}}

\begin{document}

\thispagestyle{plain}

\begin{center}
{\sc\LARGE Stationary solutions and

\vspace{.3cm}
Asymptotic flatness II}

\vspace{.3cm}
\n {Martin Reiris}\\
\vspace{.2cm}
{\small email: martin@aei.mpg.de} \\

\vspace{.1cm}
\n \textsc{Max Planck Institute f\"ur Gravitationsphysik \\ Golm - Germany}\\

\vspace{0.6cm}
\n \begin{minipage}[l]{11cm}
\begin{spacing}{1}
{\small This is the second part of the investigation started in \cite{SSI}. We prove here that Strongly Stationary ends having cubic volume growth are Weakly Asymptotically Flat. Combined with the results in \cite{SSI} this shows that Strongly Stationary ends are Asymptotically Flat with Schwarzschidian fall off. 
}
\end{spacing}

\vs
\end{minipage}
\end{center}

\section{Introduction.}
In \cite{SSI} we defined Weakly Asymptotically Flat (WAF) stationary ends, a notion generalizing as much as possible the standard one of Asymptotically Flat (AF) stationary end in General Relativity, and proved that they have to be a posteriori AF with Schwarzschidian fall off. In this second part we prove that strongly stationary ends, whose definition eliminates any a priori assumption on the asymptotic, are also WAF and therefore AF with Schwarzschidian fall off. 

To the purposes of this paper a stationary data consists of a three-manifold $M$, a Riemannian metric $g$, a twist one form $\omega$ and a positive lapse function $u$ satisfying the stationary vacuum Einstein equations 
\be\label{MEE}
\left\{
\begin{array}{l}
\ Ric=2\, \nabla \ln u \otimes \nabla \ln u +\frac{\displaystyle 2}{\displaystyle u^{4}}\, {\displaystyle \omega\otimes \omega},\\
\ \Delta \ln u=-\frac{\displaystyle 2\, \vert\, \omega\, \rvert^{2}}{\displaystyle u^{4}},\vs\\
\ {\rm div}\, \omega=4\langle \nabla \ln u, \omega \rangle,\vs\\
\ {\rm d}\, \omega=0.
\end{array}
\right.
\ee
The data $(g,\omega,u)$ arise naturally from strictly stationary vacuum space-times in General Relativity when we describe them only in terms of data in the quotient three-space. We refer to \cite{SSI} for an account on how to reconstruct the stationary vacuum space-time from $(M;g,\omega,u)$ (it is worth pointing out that $g$ is not the physical quotient metric but a conformal transformation of it \cite{SSI}). The associated space-time plays no technical role in this article and we will not refer to it anymore. The physical motivations of this research can be found in \cite{SSI}.

If the manifold $M$ is diffeomorphic to $\mathbb{R}^{3}$ minus an open ball, the metric $g$ is complete and $u$ is bounded below away from zero then $(M;g,\omega,u)$ is said to be a strongly stationary end. The condition on $u$, namely that $u(p)\geq u_{0}>0$ for all $p\in M$, plays no role in this article. From now on the manifold of strong stationary ends will be denoted by $E$. 

Let $(M,g)$ be a Riemannian manifold. Suppose that $M$ is non-compact, has compact boundary and suppose too that $g$ is complete. Then $(M,g)$ is said to have cubic volume growth if
\be\label{TENIS}
\lim_{r\rightarrow \infty} \frac{{\rm Vol}({\mathcal T}_{g}(\partial E,r))}{r^{3}}=\mu>0
\ee
where ${\mathcal T}_{g}(\partial M,r)=\big\{p\in E,\dist_{g}(p,\partial M)\leq r\big\}$ is the metric-tubular neighborhood of $\partial M$ and radius $r>0$. Note by inspecting the first equation in (\ref{MEE}) that the Ricci curvature of stationary solutions is non-negative. Therefore the quotient ${\rm Vol}\big({\mathcal T}_{g}(\partial E,r)\big)/r^{3}$ is monotonically non-increasing in $r$ by the Bishop-Gromov monotonicity and the limit (\ref{TENIS}) exists. If $\mu=0$ then $(M,g)$ is said to have less than cubic volume growth. 

The purpose of this article is then to prove,
\begin{Theorem}\label{MAIN} Let $E$ be a strongly stationary end having cubic volume growth. Then $E$ is WAF and therefore AF with Schwarzschidian fall off.
\end{Theorem}
\n The definition of WAF end is recalled in the next section after the necessary notation and terminology is introduced but before we pass into that we would like to make a couple of comments on the hypothesis of Theorem \ref{MAIN}. On one hand, as was indicated in \cite{SSI}, any strongly stationary end enjoys necessarily cubic volume growth due to quite general geometric facts [arXiv:1212.1317]. From this and Theorem \ref{MAIN} we deduce therefore that Strongly Stationary ends are always asymptotically flat with Schwarzschidian fall off (c.f. Corollary 1 in \cite{SSI}). 
On the other hand, stationary solutions with cubic volume growth and connected at infinity [\footnote{Recall that a non-compact manifold $M$ with compact boundary is said to be connected at infinity if for every compact set $K_{1}\subset M$ there is another compact set $K_{2}$ containing $K_{1}$ such that $M\setminus K_{2}$ is connected}] turn out to be diffeomorphic to $\mathbb{R}^{3}$ minus an open ball outside a compact set and therefore AF with Schwarzschidian fall off. This property can be proved by suitably adjusting the results of this article and will be discussed elsewhere.  


\subsection{Background material I.}\label{NOTATION}

We import here the material introduced in \cite{SSI} and that will be required for the technical discussions. We introduce too the most relevant terminology and notation. The definition of WAF end is given at the end.

\vs
{\sc Distance.} 

\vs
- The distance between two points $p$ and $q$ in a connected manifold $(M,g)$ is $\dist_{g}(p,q)=\inf\big\{\length_{g}(\mathscr{C}_{p,q}),\ \mathscr{C}_{p,q}\ \text{a}\ C^{1}\ \text{curve in}\ M\ \text{joining}\ p\ \text{to}\ q\big\}$. $(M,g)$ is said complete if $(M,\dist_{g})$ is complete as a metric space. 
The distance from a point $p$ to a set $\Omega\subset M$ will be denoted by $\dist_{g}(p,\Omega)=\big\{\dist_{g}(p,q),q\in \Omega\big\}$. More generally the distance between two sets $\Omega_{1}$ and $\Omega_{2}$ is denoted by $\dist_{g}(\Omega_{1},\Omega_{2})=\inf\big\{\dist_{g}(p,q),p\in\Omega_{1}, q\in \Omega_{2}\big\}$ [\footnote{Properly speaking this is not a metric in the subsets of $M$. In particular the distance is zero if for instance they share a point but are different sets.}].

- When one considers the metric induced by $g$ on a submanifold $N$ of a manifold $(M,g)$ it may become necessary to distinguish it from the restriction to $N$ of the metric induced by $g$ on $M$ (which do not necessarily coincide). When this is necessary we will use the notation $\dist_{(N,g)}$. For instance if $(N,g)\subset (M,g)$ then the diameter of $N$ with respect to the metric induced by $g$ on $N$  will be denoted by ${\rm diam}_{(N,g)}(N)=\sup \big\{\dist_{(N,g)}(p,q),p\ \text{and}\ q\ \text{in}\ N\big\}$ and called the proper diameter. 

- The metric induced on stationary ends $(E,g)$ will be noted by $\dist(p,q)$ and always without the subindex $g$. The distance function to the boundary $\partial E$ of stationary ends will be denoted with total exclusivity by $d(p)$ or simply $d$, that is, $d(p)=\dist(p,\partial E)=\inf\big\{\dist(p,q),q\in \partial E\big\}$.  

\vs
{\sc Scaling.} 

\vs
- Let $E$ be a strongly stationary end. Then, for any real number $r>0$ we will denote by $g_{r}$ to the scaled metric
\ben
g_{r}:=\frac{1}{r^{2}}\, g.
\een
Tensors and metric quantities constructed out of $g_{r}$ will be sub-indexed with an $r$. For instance, for the scalar curvature we have $R_{r}=R_{g_{r}}=R/r^{2}$ and for the Ricci curvature $Ric_{r}=Ric_{g_{r}}=Ric$ (although $Ric_{r}=Ric$ we will keep including the subindex $r$). Also, $d_{r}(p)=d(p)/r$. {\it This way of notating will be used extensively all through the article and is crucial keeping track of it.}

\vs\vs
{\sc Area, second fundamental form and mean curvature.} 

\vs
- The Riemannian-metric induced on compact embedded two-surfaces $S\subset E$ will be denoted by $h$ and the $h$-area of $S$ by $A(S)$. Following the notation introduced before, the metric induced in $S$ from $g_{r}$ is denoted by $h_{r}$ and the $h_{r}$-area of $S$, i.e. $A(S)/r^{2}$, is denoted by $A_{r}(S)$.   
The second fundamental form of $S$ (fixed some normal) will be denoted by $\Theta$ and the mean curvature ${\rm tr}_{h}\Theta$ by $\theta$. 

\vs
{\sc Annuli and metric annuli.} 

\vs
-  Let $E$ be a strongly stationary end. Then for any $0<a<b$ we will denote by ${\mathcal A}(a,b)$ (resp. ${\mathcal A}[a,b]$) the set
\ben
{\mathcal A}(a,b)=\big\{p\in E/a<d(p)<b\big\},\quad \text{(resp.}\ {\mathcal A}[a,b]=\big\{p\in E/a\leq d(p)\leq b\big\}\text{)} 
\een
and call it the open (resp. closed) metric annulus of radii $a$ and $b$. The notation ${\mathcal A}(a,b)$ (resp. ${\mathcal A}[a,b]$) will always refer to open (resp. closed) metric annuli defined with respect to the unscaled metric $g$ but
the subindex $r$ is included when the (open or closed) metric annuli are defined with respect to the scaled metric $g_{r}=g/r^{2}$, namely
\ben
{\mathcal A}_{r}(a,b)=\big\{p\in E/ a<d_{r}(p)<b\big\}\quad \text{and}\quad {\mathcal A}_{r}[a,b]=\big\{p\in E/ a\leq d_{r}(p)\leq b\big\}.
\een
This is consistent with the notation introduced before. Note that for all $r>0$ we have ${\mathcal A}(ar,br)={\mathcal A}_{r}(a,b)$ and ${\mathcal A}[ar,br]={\mathcal A}_{r}[a,b]$.

- Standard open annuli in $\mathbb{R}^{3}$ will be denoted by $\AR(a,b)$, namely,
\ben
\AR(a,b)=\big\{x\in\mathbb{R}^{3},a<|x|<b\big\} = B_{\mathbb{R}^{3}}(o,b)\setminus \overline{B_{\mathbb{R}^{3}}(o,a)}
\een
where for any $c>0$ $B_{\mathbb{R}^{3}}(o,c)$ is the open ball of center the origin $o=(0,0,0)$ and radius $c$ in $\mathbb{R}^{3}$.
As before, closed annulus in $\mathbb{R}^{3}$ are denoted by $\AR[a,b]=\big\{x\in \mathbb{R}^{3},a\leq |x|\leq b\big\}$.

- A manifold $\Omega$ is said to be an open (resp. closed) annulus if $\Omega$ is diffeomorphic to $\AR(1,2)$ (resp. $\AR[1,2]$). 
A metric annulus doesn't have to be necessarily an open annulus in this sense. In general, the shape of the metric annuli can be wild. 

\vs
{\sc Curvature.} 

\vs
- An essential property of the curvature of stationary solutions is M. T. Anderson's a priori curvature decay \cite{MR1806984}. It says that there is a universal constant ${\mathcal K}>0$ such that for any stationary solution $(M;g,\omega,u)$ and $p\in M$ we have $|Ric(p)|\leq {\mathcal K}/\dist^{2}(p,\partial M)$ [\footnote{There is a caveat here. The curvature estimate provided in {\bf Theorem 0.2} of \cite{MR1806984} is (as written) for the space-time metric and not for the metric $g$. However the proof of that Theorem is achieved by proving first the estimate $|Ric_{g}(p)|\leq {\mathcal K}/\dist^{2}_{g}(p,\partial M)$ (see c.f. {\bf Step I} in \cite{MR1806984}) that is all what we need here.}]. In strongly stationary ends $(E;g,\omega,u)$ this reads
\ben
|Ric(p)|\leq {\mathcal K}/\displaystyle d^{2}(p)
\een
for all $p\in E$. In particular for any $p\in {\mathcal A}_{r}(a,b)$, the Ricci curvature of the scaled metric $g_{r}$ is bounded as $|Ric_{r}(p)|_{r}\leq {\mathcal K}/a^{2}$. 

\vs
{\sc Norms and convergence of Riemannian manifolds.} 

\vs
- Given a tensor field $U$ (of any valence) on a region $\Omega$ of a manifold $(M,g)$, the $C^{i}_{g}$-norm of $U$ over $\Omega$ is defined as
\ben
\| U\|_{C^{i}_{g}(\Omega)}:=\sup_{p\in \Omega} \sum_{j=0}^{j=i} \big|\big(\nabla^{j} U\big)(p)\big|_{g}.
\een
Of course $\|U\|_{C^{i}_{g}(\Omega)}\leq \|U\|_{C^{i+1}_{g}(\Omega)}$. The subindex $g$ will be suppressed when $\Omega$ is a region of the Euclidean three-space, namely we will write $C^{i}$. 

All what we will need about convergence of smooth Riemannian manifolds will be restricted to the following definition (which is not the most general \cite{MR2243772}). Let $(\Omega_{m},g_{m})$ be a sequence of smooth, compact, connected three-manifolds with smooth boundary and let $(\Omega_{\infty},g_{\infty})$ be also smooth, compact, connected three-manifold with smooth boundary. Then, $(\Omega_{m},g_{m})$ converges to $(\Omega_{\infty},g_{\infty})$ in $C^{i}$, $i\geq 2$, if there are diffeomorphisms 
$\varphi_{m}:\Omega_{\infty}\rightarrow \Omega_{m}$ such that $\big\| \varphi_{m}^{*}\, g_{m} - g_{\infty}\big\|_{C^{i}_{g_{\infty}}(\Omega_{\infty})}\rightarrow 0$ where $\varphi^{*}_{m} g_{m}$ is the pull-back of $g_{m}$ by $\varphi_{m}$. The definition is the same if we do not require compactness on the $\Omega_{m}$ and $\Omega_{\infty}$ but assume uniformly bounded diameters. A sequence of smooth tensors $U_{m}$ converge to a smooth tensor $U_{\infty}$ in $C^{i}$, $i\geq 0$, if $\big\| \varphi_{m}^{*}\, U_{m} - U_{\infty}\big\|_{C^{i}_{g_{\infty}}(\Omega_{\infty})}\rightarrow 0$

\vs
{\sc WAF ends.} 

\vs
- The definition of WAF end is as follows. We refer the reader to \cite{SSI} for further comments about the definition.

\begin{Definition}\label{WAFD} A strongly stationary end $(E;g,\omega,u)$ is weakly asymptotically flat (WAF) if  
for every $i\geq 2$, $l\geq 1$ and divergent sequence $r_{m}\rightarrow \infty$, there is a sequence of open annuli $\Omega_{m}\subset E$ 
such that, 
\begin{enumerate}[labelindent=\parindent,leftmargin=*,label={\rm (W\arabic*)}]
\item ${\mathcal A}_{r_{m}}(1/2,2^{l})\subset \Omega_{m}$ for every $m$,

\item $(\Omega_{m},g_{r_{m}})$ converges in $C^{i}$ to  the flat annulus 
$\big({\mathcal A}_{\mathbb{R}^{3}}(1/2,2^{l}),g_{\mathbb{R}^{3}}\big)$,

\item The scaled distance functions $d_{r_{m}}$ (restricted to $\Omega_{m}$) converge in $C^{0}$ to the distance to the origin in $\mathbb{R}^{3}$ (restricted to ${\mathcal A}_{\mathbb{R}^{3}}(1/2,2^{l})).$

\item Every $\overline{\Omega}_{m}$ is a closed annulus and separates $\partial E$ from infinity, namely, $\partial E$ belongs to a bounded component of $E\setminus \overline{\Omega}_{m}$ for all $m$.

\end{enumerate}
\end{Definition}

\vs
The Figure \ref{Fig1} illustrates a WAF end along with some of the annuli $\Omega_{m}$.

\begin{figure}[h]
\centering
\includegraphics[width=8cm,height=6cm]{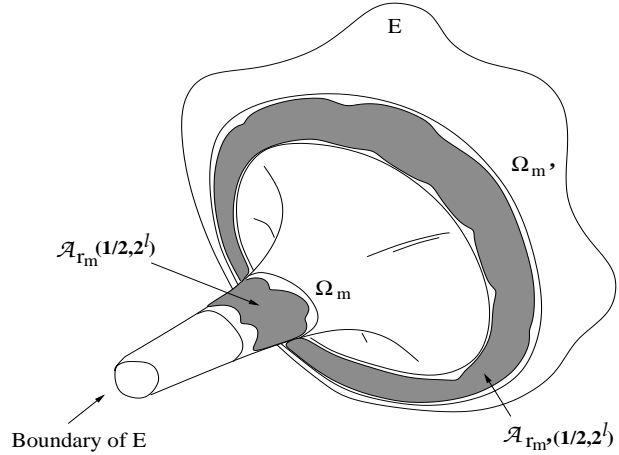}
\caption{Representation of a WAF end along with the annuli $\Omega_{m}$ and the metric annuli ${\mathcal A}_{r_{m}}(1/2,2^{l})$}
\label{Fig1}
\end{figure} 

\subsection{Background material II.}

The material contained in this section is used specifically in this article.

\vs
{\sc Regularity properties of the distance function.} 

\vs
- We summarize here quite standard properties of the distance function that have a technical relevance and will justify several operations later. The reader can consult the references for further information.

Let $M$ be a non-compact smooth manifold with compact boundary and let $g$ be a smooth complete metric with $Ric\geq 0$. Let $\dist$ be the metric induced by $g$ on $M$ and let $d$ be the distance function to $\partial M$, that is $d(p)=\dist(p,\partial M)$. 
The function $d$ is semiconcave (\cite{MR1941909}, Proposition 3.4) and therefore $\nabla d$ is locally of bounded variation (\cite{MR2041617} Theorem 2.3.1). In particular $\Delta d$ is a Radon measure and for any smooth $\phi$ of compact support in ${\rm Int}(M)$ we have $\int (\Delta d)\,\phi\, \dif V=-\int \langle \nabla \phi,\nabla d\rangle \dif V$ (note the difference in fonts between $d$ (distance) and ${\rm d}$ (differential)). By the triangle inequality the function $d$ is also 1-Lipschitz, that is $|d(p)-d(q)|\leq \dist(p,q)$.  

For every $p\in \partial M$, let $\gamma_{p}(\tau)$ be the geodesic in $M$ starting perpendicularly to $\partial M$ at $p$ (when $\tau=0$). The parameter $\tau\geq 0$ is assumed here to be the arc-length from $p$. For every $p\in \partial M$ let also $\tau_{p}=\sup\{\tau, \tau=d(\gamma_{p}(\tau))\}$. Let ${\mathcal I}$ be the subset of $\partial M\times [0,\infty)$ given by ${\mathcal I}:=\{(p,\tau),p\in \partial M\ \text{and}\ 0<\tau<\tau_{p}\}$ and consider the map
${\mathscr I}:{\mathcal I}\rightarrow M$ given by ${\mathscr I}(p,\tau)=\gamma_{p}(\tau)$. Then, the set ${\mathcal C}:=M\setminus {\mathscr I}({\mathcal I})$ is closed and of measure zero (the cut-locus) and ${\mathscr I}$ is a diffeomorphism into the image. 

As $|\nabla d|=1$ on $M\setminus {\mathcal C}$, then every the level set $\hat{\mathcal S}(\tau):=d^{-1}(\tau)\setminus {\mathcal C}$ is an embedded submanifold of $M$ of dimension two. Moreover, for almost every $\tau$ the area (${\mathcal H}^{2}$-Hausdorff measure) of $S(\tau)=d^{-1}(\tau)$ coincides with the area $A(\hat{S}(\tau))$ of $\hat{S}(\tau)$ \cite{MR2229062}. Also for almost every $\tau$ the function $\tau\rightarrow {\rm Vol}(\{p,d(p)<\tau\})$ is differentiable with $\tau$-derivative equal to $A(\hat{\mathcal S}(\tau))$. The $\tau$'s for which this holds will be called non-exceptional. 

The pull-back by ${\mathscr I}$ of the volume element in $M\setminus {\mathcal C}$ can be written as $\dif V=J \dif \tau\, \dif A_{0}$ where $\dif A_{0}$ is the area element in $\partial M$ (with the induced metric from $g$) and where $J$ is a smooth and positive function. For every $(p,\tau)\in {\mathcal I}$ we have $(\partial_{\tau} \ln J)(p,\tau)=\theta ({\mathscr I}(p,\tau))$ where $\theta({\mathscr I}(p,\tau))$ is the mean curvature of $\hat{\mathcal S}(\tau)$ at ${\mathscr I}(p,\tau)$ and in the direction of $\gamma'_{p}(\tau)$.  Also, from the focussing equation [\footnote{$\theta'=-|\Theta|^{2}-Ric(\gamma',\gamma')$.}] and the assumption $Ric\geq 0$ we have $\partial_{\tau}\theta\leq - \theta^{2}/2$. This implies easily that $\theta({\mathscr I}(p,\tau)) -2/\tau\leq 0$ for all $(p,\tau)\in {\mathcal I}$. In other words $(\partial_{\tau} \ln J -2/\tau)\leq 0$. 
From this it can be shown that the function $A(\hat{S}(\tau))/\tau^{2}$ is monotonically non-increasing in $\tau$ (although it is not necessarily continuous) \cite{MR2229062}.  

\vs
{\sc Convergence of stationary solutions} 

\vs
- The following is essentially a restatement of {\bf Lemma 1.3}  in \cite{MR1806984} with some necessary but minor modifications [\footnote{We could not validate {\bf Lemma 1.3} as it is written. I would like to thank Michael Anderson for discussions about this statement.}].

\begin{Theorem}\label{CONVSS} Let $(M_{m};g_{m},\omega_{m},u_{m})$ be a sequence of stationary solutions. 
Let $N_{m}\subset M_{m}$ be a sequence of connected open regions with compact closure and such that,
\ben
{\rm Vol}_{g_{m}}\big(N_{m}\big)\geq V_{0},\qquad {\rm {\rm diam} }_{(N_{m},g_{m})}\big(N_{m}\big)\leq D_{0},\qquad\text{and}\qquad \dist_{g_{m}}(N_{m},\partial M_{m}\big)\geq \Gamma_{0}
\een
for some $V_{0}>0$, $D_{0}<\infty$, $\Gamma_{0}>0$ and for all $m$. Then, for every $\delta<\Gamma_{0}$ there is a sequence of compact manifolds with smooth boundary $\Omega_{m}$ with $N_{m}\subset \Omega_{m}\subset {\mathcal T}_{g_{m}}(N_{m},\delta)$ such that (after scaling $\omega_{m}$ and $u_{m}$ if necessary) $(\Omega_{m};g_{m},\omega_{m},u_{m})$ has a subsequence converging in $C^{\infty}$ to a stationary solution $(\Omega_{\infty};g_{\infty},\omega_{\infty},u_{\infty})$, where $\Omega_{\infty}$ is a compact manifold with smooth boundary.
\end{Theorem}

Above ${\mathcal T}_{g_{m}}(N_{m},\gamma)=\{p\in M_{m},\dist_{g_{m}}(p,N_{m})<\gamma\}$ is the metric-tubular neighborhood of $N_{m}$ and radius $\gamma$. Note that it is the proper diameter of $N_{m}$ the one that is uniformly bounded by $D_{0}$ [\footnote{In other words ${\rm diam} _{(N_{m},g_{m})}(N_{m})\leq D_{0}$ means that for every $\varepsilon>0$ and $p$ and $q$ in $N_{m}$ there is a $C^{1}$ curve in $N_{m}$ with $g_{m}$-length less than $D_{0}+\varepsilon$.}].
The reader may find it curious that no condition on the curvature is necessary. The reason for this is that the curvature is automatically uniformly bounded on $\Omega_{m}$ by virtue of Anderson's estimate, precisely $|Ric(p)|\leq {\mathcal K}/(\Gamma_{0}-\delta)^{2}$ for any $p\in \Omega_{m}$.

\section{Proof of Theorem \ref{MAIN}.}

The proof of Theorem \ref{MAIN} is structured as follows. In Proposition \ref{RADON} we discuss a basic and general property of the Laplacian of the distance function (to the boundary) in manifolds with non-negative Ricci curvature. This is then used in Proposition \ref{PARAUSO} to study the limit (when it exists) of scalings of the distance function. The proposition is crucial to prove the central Lemma \ref{WAFCL} which, in rough terms, shows the existence of ``almost" Euclidean annuli far away from the boundary of Strongly Stationary ends having cubic volume growth. We use this lemma in Proposition \ref{AUXWAF} to study the global geometry of ends and this paves the way to prove finally in Theorem \ref{MAINTII} that Strongly Stationary ends with cubic volume growth are WAF.   

\begin{Proposition}\label{RADON} Let $(M,g)$ be a complete smooth Riemannian manifold with $Ric\geq 0$. Suppose that $M$ is non-compact and has non-empty and compact boundary. Let $d$ be the distance function to $\partial M$, that is $d(p)=\dist(p,\partial M)$. Then,
\begin{enumerate}[labelindent=\parindent,leftmargin=*,label={\rm (\roman*)}]
\item[\rm (i)] For every smooth and non-negative function $\phi$ with support in ${\rm Int}(M)$ we have 
\ben
\int_{M} \bigg[\frac{2}{d}-\big(\Delta d\big)\bigg]\, \phi\, \dif V\geq 0.
\een
In other words the Radon measure $2/d-\Delta d$ is non-negative in ${\rm Int}(M)$.
\item[\rm (ii)] For every $0<a<b$ and divergent sequence $r_{m}\rightarrow \infty$ we have  
\ben
\lim_{r_{m}\rightarrow \infty} \int_{{\mathcal A}_{r_{m}}(a,b)}\bigg[\frac{2}{d_{r_{m}}}-\big(\Delta_{r_{m}} d_{r_{m}}\big)\bigg]\, \dif V_{r_{m}} = 0.
\een
\end{enumerate}
\end{Proposition}

\begin{proof}[\bf Proof.] (i) We compute
\begin{align*}
\int_{M} \big(\Delta d\big)\, \phi\, \dif V & =-\int_{M} \langle \nabla d,\nabla \phi\rangle\, \dif V = \int_{\mathcal I} (\partial_{\tau} \phi) J\, \dif\tau\, \dif A_{0} \\ 
& = \bigg[\lim_{\nu\downarrow 0} \int_{\partial M} \big(\phi J)\big|_{(p,\tau_{p}-\nu)}\, \dif A_{0}(p)\bigg] - \int_{\mathcal I} \phi\, (\partial_{\tau} J)\, \dif\tau\, \dif A_{0} 
\end{align*}
where to pass from the second to the third integral (where we are avoiding the locus) we used that the integrand $\langle \nabla d,\nabla \phi\rangle$ is in $H^{1,2}$ and that the locus has measure zero. Then,
\ben
\int_{M} \bigg[ \frac{2}{d}-\big(\Delta d\big)\bigg]\, \phi\, \dif V= \bigg[\lim_{\nu\downarrow 0} \int_{\partial M} \big(\phi J)\big|_{(p,\tau_{p}-\nu)}\, \dif A_{0}(p)\bigg] + \int_{\mathcal I} \phi\, (\frac{2}{\tau}-\frac{\partial_{\tau} J}{J})\, J\, \dif\tau\, \dif A_{0}\geq 0
\een  
because $\phi\geq 0$, $J> 0$ and $(2/\tau -\partial_{\tau} \ln J)\geq 0$. 

(ii) Let $\tau^{+}_{m}$ and $\tau^{-}_{m}$ be two divergent sequences of non-exceptional $\tau$'s, such that 
$\tau^{+}_{r_{m}}:=\tau^{+}_{m}/r_{m}\downarrow b$ and $\tau^{-}_{r_{m}}:=\tau^{+}/r_{m}\uparrow a$. Make $\tau_{r_{m}}:=\tau/r_{m}$. Then, we compute
\begin{align}\label{LENG}
\int_{\displaystyle {\mathcal A}_{r_{m}}(\tau_{r_{m}}^{-},\tau^{+}_{r_{m}})} & \bigg[\frac{2}{d_{r_{m}}}-  \big(\Delta_{r_{m}} d_{r_{m}}  \big)\bigg]\, \dif V_{r_{m}}=  \\
\nonumber & = 2\int_{\tau^{-}_{r_{m}}}^{\tau^{+}_{r_{m}}} \bigg[\frac{A_{r_{m}}(\hat{\mathcal S}(\tau_{m}))}{\tau_{r_{m}}^{2}}\bigg]\tau_{r_{m}} \dif \tau_{r_{m}} - \bigg[A_{r_{m}}(\hat{\mathcal S}(\tau^{+}_{m}))-A_{r_{m}}(\hat{\mathcal S}(\tau^{-}_{m}))\bigg]
\end{align}
where (following the notational convention) $A_{r_{m}}(\,-\,)=A(\,-\,)/r_{m}^{2}$. Now, for every $\tau$ we have $A_{r_{m}}(\hat{\mathcal S}(\tau))/\tau_{r_{m}}^{2}=A(\hat{\mathcal S}(\tau))/\tau^{2}$ and, recall, the function $A(\hat{\mathcal S}(\tau))/\tau^{2}$ is monotonically non-increasing in $\tau$. Therefore the function $A(\hat{\mathcal S}(\tau_{r_{m}}))/\tau_{r_{m}}^{2}$ as a function of $\tau_{r_{m}}$ in the interval $[\tau^{-}_{r_{m}},\tau^{-}_{r_{m}}]$ tends to a constant, say $\mu\geq 0$, over $[a,b]$. In particular $A_{r_{m}}(\hat{\mathcal S}(\tau_{m}^{+}))-A_{r_{m}}(\hat{\mathcal S}(\tau_{m}^{-}))$ tends to $\mu (b^{2}- a^{2})$ and 
\ben
2\int_{\tau^{-}_{r_{m}}}^{\tau^{+}_{r_{m}}} \bigg[\frac{A_{r_{m}}(\hat{S}(\tau_{m}))}{\tau_{r_{m}}^{2}}\bigg]\, \tau_{r_{m}}\, \dif \tau_{r_{m}}\rightarrow \mu(b^{2}-a^{2}).
\een
As a result the right hand side of (\ref{LENG}) tends to zero as wished. \end{proof}

\begin{Proposition}\label{PARAUSO} Let $E$ be a strongly stationary end and let $r_{m}\rightarrow \infty$ be a divergent sequence.
Suppose that $(\Omega_{m},g_{r_{m}})$ converges in $C^{\infty}$ to $(\Omega_{\infty},g_{\infty})$ where the $\Omega_{m}$'s and $\Omega_{\infty}$ are compact connected manifolds with smooth boundary and where $\Omega_{m}\subset {\mathcal A}_{r_{m}}(a,b)$ for each $m$. Then, there is a subsequence such that $d_{r_{m}}$ converges in $C^{0}$ to a smooth function $d_{\infty}$ satisfying 
\ben
|\nabla d_{\infty}|_{\infty}=1\qquad \text{and}\qquad \Delta_{\infty} d_{\infty} = \frac{2}{d_{\infty}}.
\een  
\end{Proposition}

\begin{proof}[\bf Proof.] Denote by $\varphi_{m}:\Omega_{\infty}\rightarrow \Omega_{m}$ the diffeomorphisms realizing the $C^{\infty}$ convergence $(\Omega_{m},g_{r_{m}})\rightarrow (\Omega_{\infty},g_{\infty})$. Also a few times below we make reference to the 
metrics induced by $g_{r_{m}}$ on $\Omega_{m}$ and that, as we said in the introduction, will be denoted by $\dist_{(\Omega_{m},g_{r_{m}})}$. Note again that this is not the same than the distance induced by $g_{r_{m}}$ on $E$ and restricted to $\Omega_{m}$ and that we denote by $\dist_{r_m}$.

As $(\Omega_{m},g_{r_{m}})\overset{C^{\infty}}\rightarrow (\Omega_{\infty},g_{\infty})$, then the pull back of the metric functions $\dist_{(\Omega_{m},g_{r_{m}})}$, namely $\varphi^{*}_{m} \dist_{(\Omega_{m},g_{r_{m}})}=\dist_{(\Omega_{m},g_{r_{m}})}(\varphi_{m},\varphi_{m}):\Omega_{\infty}\times\Omega_{\infty}\rightarrow [0,\infty)$,
converge in $C^{0}$ to the metric function $d_{(\Omega_{\infty},g_{\infty})}:\Omega_{\infty}\times\Omega_{\infty}\rightarrow [0,\infty)$ induced by $g_{\infty}$ on $\Omega_{\infty}$. Therefore there is $m_{0}$ such that for any $m\geq m_{0}$ and $p,q$ in $\Omega_{\infty}$ we have $\dist_{(\Omega_{m},g_{r_{m}})}(\varphi_{m}(p),\varphi_{m}(q))\leq 2\dist_{(\Omega_{\infty},g_{\infty})}(p,q)$.
Now, for $m\geq m_{0}$ we have
\begin{align}\label{EQQ}
|d_{r_{m}} (\varphi_{m}(p))- d_{r_{m}}(\varphi_{m}(q))| & \leq \dist_{r_{m}}(\varphi_{m}(p),\varphi_{m}(q)) \\ \nonumber & \leq \dist_{(\Omega_{m},g_{r_{m}})}(\varphi_{m}(p),\varphi_{m}(q)) \leq 2\dist_{(\Omega_{\infty},g_{\infty})}(p,q)
\end{align} 
where the first inequality is just the triangle inequality.
Moreover, for all $m$ we have $|d_{r_{m}}\circ \varphi_{m}|\leq b$ because $\Omega_{m}\subset {\mathcal A}_{r_{m}}(a,b)$. This shows that the sequence of functions $\{d_{r_{m}}\circ \varphi_{m}\}_{m\geq m_{0}}$, as functions in the compact metric space $(\Omega_{\infty},\dist_{(\Omega_{\infty},g_{\infty})})$ are uniformly bounded and 2-Lipschitz (and therefore equicontinuous). By Ascoli-Arzel\` a there is a subsequence converging in $C^{0}$ to a Lipschitz function that we will denote by $d_{\infty}$. The limit function $d_{\infty}$ is indeed $1$-Lipschitz, that is $|d_{\infty}(p)-d_{\infty}(q)|\leq \dist_{(\Omega_{\infty},g_{\infty})}(p,q)$, as can be seen by taking the limit in the first and third terms of (\ref{EQQ}). During the rest of the proof we will work with such subsequence (indexed by $m$ again) and the limit function $d_{\infty}$.

We claim that for any smooth function $\phi$ of compact support in ${\rm Int}(\Omega_{\infty})$ we have
\be\label{WCOND}
\int_{\Omega_{\infty}} \bigg[\big(\Delta_{\infty} \phi\big) d_{\infty} -\big(\frac{2}{d_{\infty}}\big)\phi \bigg]\, \dif V_{\infty} = 0.
\ee
By proving the claim one would be showing that $f=d_{\infty}$ is a weak solution of $\Delta_{\infty} f= 2/d_{\infty}$ \cite{MR737190}, where we think here the right hand side as a given Lipschitz function. From the regularity of weak solutions \cite{MR737190} $f$ would then be in $H^{2,2}$. But if a positive function $f$ is in $H^{2,2}$ and satisfies $\Delta_{\infty} f=2/f$ then $f$ is smooth by a standard bootstrap of regularity. The smoothness of $d_{\infty}$ would thus follow from proving the claim. 

To see (\ref{WCOND}) for every $\phi$ we proceed as follows. First observe that it is enough to prove (\ref{WCOND}) for any $\phi\geq 0$ of compact support in ${\rm Int}(\Omega_{\infty})$ because any $\phi$ of compact support can be written as $\phi=\phi^{+}_{1}-\phi^{+}_{2}$ with $\phi^{+}_{1}\geq 0$ and $\phi^{+}_{2}\geq 0$ and of compact support [\footnote{To see this chose any non-negative function $\tilde{\phi}$ of compact support that takes the value $\sup\{|\phi|\}$ all over the support of $\phi$. Then if we let $\phi^{+}_{1}=\tilde{\phi}$ and $\phi^{+}_{2}=\tilde{\phi}-\phi$, then $\phi^{+}_{1}$ and $\phi^{+}_{2}$ are non-negative, have compact support and their subtraction is $\phi$.}].  
Assume then that $\phi\geq 0$. In $\Omega_{m}$ define the function $\phi_{m}:=\phi\circ \varphi_{m}^{-1}$ and let $\overline{\phi}=\max\{\phi\}=\max\{\phi_{m}\}$. Then,
\begin{align*}
\int_{\Omega_{\infty}} & \bigg[\big(\frac{2}{d_{\infty}}\big)\phi-\big(\Delta_{\infty} \phi\big) d_{\infty} \bigg]\, \dif V_{\infty} = \lim_{m} \int_{\Omega_{m}}  
\bigg[\big(\frac{2}{d_{r_{m}}}\big)\phi_{m} - \big(\Delta_{r_{m}} \phi_{m}\big) d_{r_{m}}\bigg]\, \dif V_{r_{m}}\\
& =\lim_{m} \int_{\Omega_{m}}  
\bigg[\frac{2}{d_{r_{m}}} - \big(\Delta_{r_{m}} d_{r_{m}}\big) \bigg]\, \phi_{m}\, \dif V_{r_{m}}\leq \overline{\phi} \, \lim_{m} \int_{\Omega_{m}} \bigg[\frac{2}{d_{r_{m}}} - \big(\Delta_{r_{m}} d_{r_{m}}\big) \bigg]\, \dif V_{r_{m}} \\
& \leq \overline{\phi} \, \lim_{m} \int_{{\mathcal A}_{r_{m}}(a,b)} \bigg[\frac{2}{d_{r_{m}}} - \big(\Delta_{r_{m}} d_{r_{m}}\big) \bigg]\, \dif V_{r_{m}} = 0
\end{align*} 
where to pass from the third to the fourth term and also from the fourth to the fifth we used (i) in Proposition \ref{RADON} and where to obtain the last equality we used (ii) in the same Proposition. To conclude that the first integral is indeed zero (and not negative), observe that it is equal to the third term which is non-negative by (i) in Proposition \ref{RADON}.

It remains to prove that $|\nabla d_{\infty}|_{\infty}=1$. Indeed, as $d_{\infty}$ is $1$-Lipschitz we have at least $|\nabla d_{\infty}|_{\infty}\leq 1$. To show that the norm is indeed one it is enough to prove that: {\it For any $p\in {\rm Int}(\Omega_{\infty})$ there is $\varepsilon_{p}$ such that for any $\varepsilon<\varepsilon_{p}$ there is $q^{\varepsilon}$ such that}
\ben
d_{\infty}(p)-d_{\infty}(q^{\varepsilon})=\dist_{(\Omega_{\infty},g_{\infty})}(p,q^{\varepsilon})=\varepsilon.
\een  
Let us see this now. 
Let $p_{m}=\varphi_{m}(p)$ and for every $m$ let $\gamma_{p_{m}}(\tau)$ be a geodesic joining $p_{m}$ to $\partial E$ such that $\tau=d_{r_{m}}(\gamma_{p_{m}}(\tau))$ for all $0\leq \tau\leq d_{r_{m}}(p_{m})$. Such geodesic must minimize the distance between any two of its points. Therefore, if for any $\varepsilon<d_{r_{m}}(p_{m})$ we let $q^{\varepsilon}_{m}:=\gamma_{p_{m}}(d_{r_{m}}(p_{m})-\varepsilon)$ then we have $d_{r_{m}}(p_{m})-d_{r_{m}}(q^{\varepsilon}_{m})=\dist_{rm}(p_{m},q^{\varepsilon}_{m})=\varepsilon$. 
Now, if $\varepsilon\leq \varepsilon_{p}=\dist_{(\Omega_{\infty},g_{\infty})}(p,\partial \Omega_{\infty})/2$ then there is $m_{\varepsilon}$ such that for any $m\geq m_{\varepsilon}$ we have $q_{m}^{\varepsilon}\in \Omega_{m}$ and $\dist_{r_{m}}(p_{m},q_{m}^{\varepsilon})=d_{(\Omega_{m},g_{r_{m}})}(p_{m},q_{m}^{\varepsilon})$. Therefore, one can take a subsequence of $\{\varphi_{m}^{-1}(q_{m}^{\varepsilon})\}_{m\geq m_{\varepsilon}}$ (indexed again by $m$) and converging to a $q^{\varepsilon}$ satisfying 
\begin{align*}
d_{\infty}(p)-d_{\infty}(q^{\varepsilon}) & =\lim\big(d_{r_{m}}(p_{m})-d_{r_{m}}(q^{\varepsilon}_{m})\big)=\lim \dist_{rm}(p_{m},q^{\varepsilon}_{m})\\
& = \lim \dist_{(\Omega_{m},g_{r_{m}})}(p_{m},q_{m}^{\varepsilon})= \dist_{(\Omega_{\infty},g_{\infty})}(p,q^{\varepsilon})=\varepsilon
\end{align*}
as wished. \end{proof}
\begin{Lemma}\label{WAFCL} Let $E$ be a strong stationary end having cubic volume growth. Then, for every $V>0$, $\varepsilon>0$, integer $i\geq 2$ and $b>a>0$ there is $r_{0}=r_{0}(V,\varepsilon, a, b,i)>0$ such that for every $r\geq r_{0}$ and every open and connected region ${\mathcal U}$ with
\be\label{CONDFI}
{\mathcal U}\subset {\mathcal A}_{r}(a,b)\qquad \text{and}\qquad {\rm Vol}_{r}({\mathcal U})\geq V,
\ee
there exists a closed annulus ${\mathcal W}$ with ${\mathcal A}_{r}(a/2,2b)\supset {\mathcal W}\supset {\mathcal U}$ and a diffeomorphism $\varphi:\AR[2a/3,3b/2]\rightarrow {\mathcal W}$ satisfying simultaneously 
\begin{enumerate}[labelindent=\parindent,leftmargin=*,label={\rm (P\arabic*)}]
\item $\varphi\big(\partial B_{\mathbb{R}^{3}}(o,2a/3)\big) \subset {\mathcal A}_{r}(a/2,a)$ and $\varphi\big(\partial B_{\mathbb{R}^{3}}(o,3b/2)\big) \subset {\mathcal A}_{r}(b,2b)$,
\item $\varphi^{*} (g_{r})$ is $\varepsilon$-close in the $C^{i}$-norm to the Euclidean metric in $\AR[2a/3,3b/2]$,
\item $d_{r}\circ \varphi$ is $\varepsilon$-close in the $C^{0}$-norm to the distance function to the origin in $\mathbb{R}^{3}$ restricted to $\AR[2a/3,3b/2]$. 
\end{enumerate}
\end{Lemma}
\begin{proof}[\bf Proof.] The proof proceeds by contradiction. Assume then that there exists $V>0$, $\varepsilon>0$,  $i\geq 2$, $0< a<b$, a divergent sequence $r_{m}\rightarrow \infty$ and a sequence of connected open regions ${\mathcal U}_{m}$ satisfying
\be\label{UCONDI}
{\mathcal U}_{m}\subset {\mathcal A}_{r_{m}}(a,b)\qquad\text{and}\qquad {\rm Vol}_{r_{m}}({\mathcal U}_{m})\geq V
\ee
for ever $m$, but such that (also for each $m$) there does not exist a closed annulus ${\mathcal W}_{m}$, with ${\mathcal U}_{m}\subset {\mathcal W}_{m}\subset {\mathcal A}_{r_{m}}(a/2,2b)$, together with a diffeomorphism $\varphi_{m}:\AR[2a/3,3b/2]\rightarrow {\mathcal W}_{m}$ satisfying simultaneously,
\begin{enumerate}[labelindent=\parindent,leftmargin=*,label={\rm (P\arabic*')}]
\item $\varphi_{m}(\partial B_{\mathbb{R}^{3}}(o,2a/3))\subset {\mathcal A}_{r}(a/2,a)$, $\varphi_{m}(\partial B_{\mathbb{R}^{3}}(o,3b/2))\subset {\mathcal A}_{r}(b,2b)$, 
\item $\varphi_{m}^{*} g_{r_{m}}$ is $\varepsilon$-close in the $C^{2}$-norm to the Euclidean metric, 
\item $d_{r_{m}}\circ \varphi_{m}$ is $\varepsilon$-close in the $C^{0}$-norm to the distance function to the origin in $\mathbb{R}^{3}$. 
\end{enumerate}

\n We will see in what follows that for sufficiently large $m$ a region ${\mathcal W}_{m}$ with ${\mathcal U}_{m}\subset {\mathcal W}_{m}\subset {\mathcal A}_{r_{m}}(a/2,2b)$ and a diffeomorphism $\varphi_{m}$ can indeed be found satisfying (P1')-(P3'). In this way a contradiction will be reached. 

First, by Liu's Ball-Covering-Property (c.f. {\it Remark 2} \cite{MR1068127} with $S {\rm (there)} = {\mathcal A}_{r_{m}}[a/4,4b]$ {\rm (here)} and $
\mu{\rm (there)} = a/16 {\rm (here)}$ [\footnote{There is a caveat in this point. To define (here) the analogous to the point $p_{0}$ (there) from which distances are measured proceed as follows. ``Fill in" smoothly $E$ by gluing a three-ball $B_{\mathbb{R}^{3}}(o,1)$ and provide the ball with a Riemannian metric $g$ in such a way every point $x\in \partial B_{\mathbb{R}^{3}}(o,1)$ is at a $g$-distance one from the origin $o$ (this can always be done). Then for any $p\in E$ we have $d(p)=\dist_{(E\cup_{\sim} B_{\mathbb{R}^{3}}(o,1)),g)}(p,o)-1$. In this setup, when using {\it Remark 2}, make it with $M{\rm (there)}=E\cup_{\sim} B_{\mathbb{R}^{3}}(o,1){\rm (here)}$ and $p_{0}{\rm (there)}=o{\rm (here)}$.}]), there is an integer ${\mathfrak N}>0$ such that for each $m$ there are geodesic balls $B_{g_{r_{m}}}(p_{m,j},a/16)$, $j=1,\ldots,j_{m}\leq {\mathfrak N}$, each of which intersects ${\mathcal A}_{r_{m}}[a/4,4b]$ and the union of which covers ${\mathcal A}_{r_{m}}[a/4,4b]$. What is crucial here is that the bound ${\mathfrak N}$ for the number of balls is independent of $m$. 
For each $m$ let ${\mathcal B}_{m}$ be the connected component of the union $\cup_{j=1}^{j=j_{m}} B_{g_{r_{m}}}(p_{m,j},a/16)$ containing the connected set ${\mathcal U}_{m}$. We claim that for each $m$ we have
\begin{enumerate}[labelindent=\parindent,leftmargin=*,label={\rm (Q\arabic*)}]
\item [\rm (a)] ${\rm Vol}_{r_{m}}({\mathcal B}_{m})\geq V$, and
\item [\rm (b)] ${\rm diam} _{({\mathcal B}_{m},g_{r_{m}})}({\mathcal B}_{m})\leq a{\mathfrak N}/8$, and
\item [\rm (c)] $\dist_{r_{m}}\big({\mathcal B}_{m},\partial {\mathcal A}_{r_{m}}(a/16,16b)\big)\geq a/16$.
\end{enumerate}
\n Indeed, (a) follows from (\ref{UCONDI}) and from the inclusion ${\mathcal U}_{m}\subset {\mathcal B}_{m}$; (b) follows from the general geometric fact that every connected set which is the union of $N$ geodesic balls of radii $D$ has a proper diameter of at most $2DN$; 
(c) To show this we note first that
${\mathcal B}_{m}\subset {\mathcal A}_{r}(a/8,8b)$. Indeed, if $p\in {\mathcal B}_{m}$ then it belongs to a geodesic ball of $g_{r_{m}}$-radius $a/16$ intersecting ${\mathcal A}_{r}[a/4,4b]$. Thus there is a point $q$ with $a/4\leq d_{r_{m}}(q)\leq 4b$ such that $\dist_{r_{m}}(p,q)<a/8$ (i.e. twice the radius). Then by the triangle inequality we have $d_{r_{m}}(p)\geq d_{r_{m}}(q)-\dist_{r_{m}}(p,q)>a/4-a/8=a/8$ and $d_{r_{m}}(p)\leq d_{r_{m}}(q)+\dist_{r_{m}}(p,q)<4b+a/8=8b$ as wished. 
On the other hand if a point $p'$ is in $\partial {\mathcal A}_{r_{m}}(a/16,16b)$ then we have either (i) $d_{r_{m}}(p')=a/16$, or (ii) $d_{r_{m}}(p')=16b$. Hence for any $q'\in {\mathcal A}_{r_{m}}(a/8,8b)$ that is with $a/8< d_{r_{m}}(q')< 8b$ we have, in case (i), $\dist_{r_{m}}(p',q')\geq d_{r_{m}}(q')-d_{r_{m}}(p')\geq a/8-a/16=a/16$ and, in case (ii), $\dist_{r_{m}}(p',q')\geq d_{r_{m}}(p')-d_{r_{m}}(q')\geq 16b-8b=8b>a/16$.
Thus, $\dist_{r_{m}}({\mathcal A}_{r_{m}}(a/8,8b),\partial {\mathcal A}_{r_{m}}(a/16,16b))\geq a/16$. As ${\mathcal B}_{m}\subset {\mathcal A}_{r_{m}}(a/8,8b)$ we obtain (c).

\vs
We can then use Theorem \ref{CONVSS} with $(M_{m};g_{m},\omega_{m},u_{m})=({\mathcal A}_{r_{m}}(a/16,16b);g_{r_{m}},\omega,u)$,
$N_{m}={\mathcal B}_{m}$ and $\delta=a/32$ to conclude that there is a sequence of compact manifolds with boundary $\Omega_{m}$ with ${\mathcal B}_{m}\subset \Omega_{m}\subset {\mathcal T}_{g_{r_{m}}}({\mathcal B}_{m},a/32)$, such that, after scalings $\omega_{m}:=\lambda^{2}_{m}\omega$ and $u_{m}:=\lambda_{m}u$ if necessary, there is a subsequence of $(\Omega_{m};g_{r_{m}},\omega_{m},u_{m})$ converging in $C^{\infty}$ to a stationary solution $(\Omega_{\infty};g_{\infty},\omega_{\infty},u_{\infty})$. 
By Proposition \ref{PARAUSO} one can take a further subsequence for which the distance functions $d_{r_{m}}$ converge in $C^{0}$ to a smooth function $d_{\infty}$ with $|\nabla d_{\infty}|_{\infty}=1$ and $\Delta_{\infty} d_{\infty} =2/d_{\infty}$. We will use this function $d_{\infty}$ below. 

We claim that for any $p\in\partial \Omega_{m}$ we have either $d_{r_{m}}(p)\leq (13/32)a$ or $d_{r_{m}}(p)\geq 3b$. Let us see this claim now. Let $p\in \partial \Omega_{m}$. Then the distance from $p$ to ${\mathcal B}_{m}$ is less than $a/32$ and, because ${\mathcal A}_{r_{m}}[a/4,4b]$ is a compact inside the open set ${\mathcal B}_{m}$, there must be a point $q$ in ${\mathcal B}_{m}\setminus {\mathcal A}_{r_{m}}[a/4,4b]$ such that $\dist_{r_{m}}(p,q)<a/32$. The point $q$ then belongs to a ball of $g_{r_{m}}$-radius $a/16$ intersecting ${\mathcal A}_{r_{m}}[a/4,4b]$ and therefore there must be a point $q'$ in the same ball having either (i) $d_{r_{m}}(q')=a/4$ or (ii) $d_{r_{m}}(q')=4b$. In case (i) we compute $d_{r_{m}}(p)\leq \dist_{r_{m}}(p,q)+\dist_{r_{m}}(q,q')+d_{r_{m}}(q')\leq a/32+a/8+a/4=(13/32)a$, and in case (ii) we compute $d_{r_{m}}(p)\geq d_{r_{m}}(q')-\dist_{r_{m}}(q',q)-\dist_{r_{m}}(q,p)\geq 4b-a/8-a/32>3b$. 

As a consequence for every $p\in \partial \Omega_{\infty}$ we have either $d_{\infty}(p)\leq 13a/32$ or $d_{\infty}(p)\geq 3b$. Therefore 
as $[3a/7,7b/3]\subset (13a/32,3b)$ then for every $\tau$ with $3a/7\leq \tau\leq 7b/3$ the set $d^{-1}_{\infty}(\tau)$ is compact in ${\rm Int}(\Omega_{\infty})$. Also, as $|\nabla d_{\infty}|_{\infty}=1$, every $\tau\in [3a/7,7b/3]$ is a regular value of $d_{\infty}$ and therefore $d_{\infty}^{-1}(\tau)$ is a finite union of compact and boundary-less manifolds.

Take a sequence $p_{m}\in {\mathcal U}_{m}\subset {\mathcal B}_{m}$ and suppose (restricting to a subsequence if necessary) that $p_{m}$ converges to a point $p_{\infty}$. Because of (\ref{UCONDI}) we have $a\leq d_{r_{m}}(p_{m})\leq b$ for every $m$ and therefore $a\leq d_{\infty}(p_{\infty})\leq b$. 
Denote by $\beta(t)$ the integral curve of the vector field $\nor:=\nabla d_{\infty}$ passing through $p_{\infty}$ (to simplify notation we make $\nor=\nabla d_{\infty}$ from now on).
As $|\nor|_{\infty}=1$, then for every $t_{1}<t_{2}$ we have $d_{\infty}(\beta(t_{2}))-d_{\infty}(\beta(t_{1}))=t_{2}-t_{1}$. Thus, $\beta(t)$ must reach the boundary of $\Omega_{\infty}$ at two different times (otherwise $d_{\infty}$ could get $-\infty$ and $+\infty$). For this reason, the range of $d_{\infty}(\beta(t))$ must contain the interval $[3a/7,7b/3]$. 
Also, for every $\tau\in [3a/7,7b/3]$ there is a unique $t$ such that $\tau=d_{\infty}(\beta(t))$ and we can consider the component of $d^{-1}_{\infty}(\tau)$ containing $\beta(t)$ that we will denote by ${\mathcal S}(\tau)$. 
Note, to be used below, that any two ${\mathcal S}(\tau_{1})$ and ${\mathcal S}(\tau_{2})$ ($\tau_{1}$ and $\tau_{2}$ in $[3a/7,7b/3]$) are naturally identified by the unique diffeomorphism $\phi_{\tau_{1},\tau_{2}}:{\mathcal S}(\tau_{1})\rightarrow {\mathcal S}(\tau_{2})$ defined as: $\phi_{\tau_{1},\tau_{2}}(p_{1})=p_{2}$ iff  the integral curve of $\nor$ passing through $p_{1}$ also passes through $p_{2}$. In other words ${\mathcal S}(\tau_{2})$ is identified to ${\mathcal S}(\tau_{1})$ by ``flowing" ${\mathcal S}(\tau_{1})$ through $\nor$ a parametric time equal to $\tau_{2}-\tau_{1}$.

Make ${\mathscr A}:=\cup_{\tau\in [3a/7,7b/3]} {\mathcal S}(\tau)$. We claim that every ${\mathcal S}(\tau)\subset {\mathscr A}$ is a sphere and that the induced Riemannian-metric, denoted here by $h^{\tau}$, is round and of Gaussian curvature $\kappa^{\tau}=1/\tau^{2}$.  
Moreover we also claim that the second fundamental form $\Theta^{\tau}$ of ${\mathcal S}(\tau)\subset {\mathscr A}$ (in the direction of $\nor$) is $\Theta^{\tau}=\tau h^{\tau}$. Let us prove the claim now. Everywhere in what follows we assume ${\mathcal S}(\tau)\subset {\mathscr A}$. First observe that the mean curvature $\theta^{\tau}(p)$ at a point $p\in {\mathcal S}(\tau)$ is calculated as $\theta^{\tau}(p)=\big(div_{\infty} \nor\big)(p)=\big(\Delta_{\infty} d_{\infty}\big)(p)=2/d_{\infty}(p)=2/\tau$. Then observe that the evolution of the mean curvature $\theta^{\tau}$ along any integral curve of $\nor$ is [\footnote{The equation (\ref{FOCDOS}) is a general evolution equation holding every time we have $|\nabla d_{\infty}|_{\infty}=1$. But, as a matter of fact, every integral line of $\nor=\nabla d_{\infty}$ is a geodesic (this is easily deduced easily from $|\nabla d_{\infty}|_{\infty}=1$) and (\ref{FOCDOS}) is just the focussing equation.}]
\be\label{FOCDOS}
\partial_{\tau}\theta^{\tau}=-|\Theta^{\tau}|^{2}_{h^{\tau}}-Ric_{\infty}(\nor,\nor)=-\frac{(\theta^{\tau})^{2}}{2}-|\widehat{\Theta}^{\tau}|_{h^{\tau}}^{2}-Ric_{\infty}(\nor,\nor)
\ee
and because $\theta^{\tau}=2/\tau^{2}$ and $Ric_{\infty}\geq 0$ we obtain $Ric_{\infty}(\nor,\nor)=0$ and $\widehat{\Theta}^{\tau}=0$. Hence $\Theta^{\tau}=\tau h^{\tau}$ on each ${\mathcal S}(\tau)$ (here $Ric_{\infty}=Ric_{g_{\infty}}$). Moreover from the Gauss-Codazzi equation we obtain, at each ${\mathcal S}(\tau)$,
\ben
2\kappa^{\tau}=-|\Theta^{\tau}|^{2}_{h^{\tau}}+(\theta^{\tau})^{2}+R_{\infty}-2Ric_{\infty}(\nor,\nor)=\frac{2}{\tau^{2}}+R_{\infty}
\een
where $R_{\infty}$ is the $g_{\infty}$ scalar curvature. Therefore $\kappa^{\tau}>0$ for all $\tau\in [3a/7,7b/3]$. This implies that every ${\mathcal S}(\tau)$ is a two-sphere. In addition, we would have $\kappa^{\tau}=1/\tau$ on every ${\mathcal S}(\tau)$ as long as $R_{\infty}=0$ all over ${\mathscr A}$. To see this we observe first that from the first equation in (\ref{MEE}) we have
$0=Ric_{\infty}(\nor,\nor)\geq 2 (\nor (u_{\infty}))^{2}/u_{\infty}^{2}$ and therefore $\nor(u_{\infty})=0$ all over ${\mathscr A}$. 
Because of this and because $\nor$ is perpendicular to every ${\mathcal S}(\tau)$ the integral of $\Delta_{\infty}\ln u_{\infty}= -2|\omega_{\infty}|^{2}_{\infty}/u_{\infty}^{4}$ in ${\mathscr A}$ is zero (recall ${\mathscr A}$ is the region enclosed by the two surfaces ${\mathcal S}(3a/7)$ and ${\mathcal S}(7b/3)$). Hence $\omega_{\infty}$ is identically zero in ${\mathscr A}$. Thus we have $\Delta_{\infty}\ln u_{\infty}=0$ in ${\mathscr A}$. Again, multiplying this by $\ln u_{\infty}$ and integrating gives $\int_{\mathscr A} |\nabla \ln u_{\infty}|_{\infty} {\rm d}V_{\infty}=0$. Hence $u_{\infty}$ is a constant all over ${\mathscr A}$. Finally from the first Einstein equation in (\ref{MEE}) we deduce that $Ric_{\infty}=0$ and therefore that $R_{\infty}=0$ as wished. 

Define now a diffeomorphism $\phi:\AR[3a/7,7b/3]\rightarrow {\mathscr A}$ as follows. Fix an isommetry $\psi$ from the unit sphere in Euclidean three-space into ${\mathcal S}(\tau=1)$. Then for any $x\in {\mathcal A}_{\mathbb{R}^{3}}[3a/7,7b/3]$ define $\phi(x)=\phi_{\tau_{1}=1,\tau_{2}=|x|}(x/|x|)$ where $\phi_{\tau_{1}=1,\tau_{2}=|x|}:{\mathcal S}(\tau_{1}=1)\rightarrow {\mathcal S}(\tau_{2}=|x|)$ is the diffeomorphism introduced before. One directly checks that the map $\phi$ is an isometry from $\big(\AR[3a/7,7b/3],g_{\mathbb{R}^{3}}\big)$ into $\big({\mathscr A},g_{\infty}\big)$. Moreover $(d_{\infty}\circ \phi) (x)=|x|$, that is, the pull back of $d_{\infty}$ by $\phi$ is the distance function to the origin in the Euclidean three-space. 
Therefore, the annulus ${\mathcal W}_{m}:=\tilde{\varphi}_{m}\big(\phi(
\AR(a/2,2b)\big)$ together with the diffeomorphism $\varphi_{m}=\tilde{\varphi}_{m}\circ \phi$ verify (P1')-(P3') for $m$ sufficiently large. We get thus the desired contradiction.
\end{proof}

\vs
\begin{Proposition}\label{AUXWAF} Let $E$ be a strongly stationary end having cubic volume growth. Then, there is $\hat{r}_{0}>0$ and a sequence $\{{\mathscr A}_{j}\}_{j\geq 0}$ of closed annuli in $E$ such that if we make $\hat{r}_{j}=2^{j}\hat{r}_{0}$, $j=0,1,2,\ldots$, then,
\begin{enumerate}[labelindent=\parindent,leftmargin=*,label={\rm (Q\arabic*)}]
\item ${\mathscr A}_{j}\subset {\mathcal A}_{\hat{r}_{j}}(1/2,8)$ for every $j\geq 0$, and in addition for every $j\geq 1$ one of the two spheres of $\partial {\mathscr A}_{j}$ lies in ${\mathcal A}_{\hat{r}_{j}}(1/2,1)\cap {\mathscr A}_{j-1}$ while the other lies in ${\mathcal A}_{\hat{r}_{j}}(4,8)\cap {\mathscr A}_{j+1}$. 
\item Every finite union $\cup_{j=J_{1}}^{j=J_{2}} {\mathscr A}_{j}$, with $J_{2}\geq J_{1}$, is diffeomorphic to the annulus $\AR[1,2]$ and the infinite union $\cup_{j=0}^{j=\infty} {\mathscr A}_{j}$ covers $E$ up to a set of compact closure.
\end{enumerate}
\end{Proposition}

\vs
The proof of the Proposition \ref{AUXWAF} requires some preparation. Define $V_{0}$ as
\ben
V_{0}=\frac{{\rm Vol}_{g_{\mathbb{R}^{3}}}\big(\AR(4/3,5/3)\big)}{2}
\een
that is, as one half of the volume of the annulus $\AR(4/3,5/3)$ which, observe, is roughly speaking the central ``third" of the annulus $\AR(1,2)$.
Now, let $\varepsilon_{0}>0$ be small enough such that for any $\hat{r}$ (but no matter which) and for any diffeomorphism $\varphi:\AR[2/3,6]\rightarrow {\mathcal W}\subset {\mathcal A}_{\hat{r}}(1/2,8)$ satisfying (P1)-(P3) with $\varepsilon=\varepsilon_{0}$, $a=1$, $b=4$, $i=2$ and $r=\hat{r}$, then we have
\be
\label{VOLCON} 
\varphi\big(\AR(8/3,10/3)\big) \subset {\mathcal A}_{\hat{r}}(2,4),\qquad \text{and}\qquad {\rm Vol}_{2\hat{r}}\big(\varphi\big(\AR(8/3,10/3)\big)\big)\geq V_{0}.
\ee
Note that the annulus $\AR(8/3,10/3)$ is ``$\AR(4/3,5/3)$ magnified by a factor of two" and is roughly speaking the central ``third" of the annulus $\AR(2,4)$. 
If we now let ${\mathcal U}=\varphi(\AR(8/3,10/3))$ then (\ref{VOLCON}) is the same as 
\ben
{\mathcal U}\subset {\mathcal A}_{2\hat{r}}(1,2)\qquad \text{and}\qquad {\rm Vol}_{2\hat{r}}({\mathcal U})\geq V_{0}.
\een
In other words the conditions (\ref{CONDFI}) in Lemma \ref{WAFCL} with $V=V_{0}$, $a=1$, $b=4$ and $r=2\hat{r}$ will be satisfied. This fact will be used repeatedly in the proof of Proposition \ref{AUXWAF}.

The following proposition will help to start the iteration in the proof of the Proposition \ref{AUXWAF}. In the statement below we let $r_{0}:=r_{0}(V=V_{0},\varepsilon=\varepsilon_{0},a=1,b=4,i=2)$, namely the $r_{0}$ provided by the Lemma \ref{WAFCL} with $V=V_{0},\varepsilon=\varepsilon_{0},a=1,b=4$ and $i=2$.

\begin{Proposition}\label{AUXX} Let $E$ be a strongly stationary end having cubic volume growth. Let $V_{0}>0$ and $r_{0}$ be as defined before. Then, there is $\hat{r}_{0}\geq r_{0}$ and an open and connected region $\hat{\mathcal U}_{0}$ such that 
\ben
\hat{\mathcal U}_{0}\subset {\mathcal A}_{\hat{r}_{0}}(1,4)\quad\text{and}\quad {\rm Vol}_{\hat{r}_{0}}(\hat{\mathcal U}_{0})\geq V_{0}.
\een
\end{Proposition}

\vs
As the reader will see the proposition is valid for any $V_{0}$ and $r_{0}$ and not just the ones specified before. Nevertheless it will be only used with the values signaled.  

\n \begin{proof}[\bf Proof.] By the Bishop-Gromov monotonicity, the quotient ${\rm Vol}\big({\mathcal T}_{g}(\partial E,r)\big)/r^{3}$ is monotonically non-increasing in $r$ and by the assumption of cubic volume growth the limit is non zero, say it is $\mu>0$. Then $\lim {\rm Vol}_{r} \big({\mathcal A}_{r}(2,3)\big)=(3^{3}-2^{3})\mu=19\mu$. Let $r_{m}\rightarrow \infty$ be an arbitrary divergent sequence. As $lim_{m} {\rm Vol}_{r_{m}}\big({\mathcal A}_{r_{m}}(2,3)\big)=19\mu>0$ we can assume ${\rm Vol}_{r_{m}}\big({\mathcal A}_{r_{m}}(2,3)\big)\geq \mu_{1}$ for some $\mu_{1}>0$ and for all $m$. 
By Liu's ball covering property \cite{MR1068127} there is an integer ${\mathfrak N}>0$ such that for every $m$ there is a set of geodesic balls $B_{g_{r_{m}}}(p_{m,j},1/4),j=1,\dots,j_{m}\leq {\mathfrak N}$ each of which intersects ${\mathcal A}_{r_{m}}[2,3]$ and the union of which covers ${\mathcal A}_{r_{m}}[2,3]$. Each ball is inside ${\mathcal A}_{r_{m}}(1,4)$ (this is simple to see) and there must be necessarily one (say the ball $B_{g_{r_{m}}}(p_{m,1},1/4)$ if they are ordered appropriately) with $g_{r_{m}}$-volume greater or equal than $\mu_{1}/{\mathfrak N}$. Define $V_{1}:=\mu_{1}/{\mathfrak N}$ and ${\mathcal U}_{m}:=B_{g_{r_{m}}}(p_{m,1},1/4)$. Then, for all $m$, we have
\ben
{\mathcal U}_{m}\subset {\mathcal A}_{r_{m}}(1,4)\quad\text{and}\quad {\rm Vol}_{r_{m}}({\mathcal U}_{m})\geq V_{1}
\een
Now, for every integer $k\geq 1$ let $r_{0}(k):=r_{0}(V=V_{1},\epsilon=1/k,a=1,b=4,i=2)$, i.e. the value of $r_{0}$ provided by Lemma \ref{WAFCL} when $V=V_{1}$, $\epsilon=1/k$, $a=1$, $b=4$ and $i=2$. 
Also, for every $k\geq 1$ let $m(k)$ be any $m$ for which $r_{m(k)}> r_{0}(k)$. Then, according to Lemma \ref{WAFCL}, (used with $V=V_{1}, \epsilon=1/k, a=1, b=4, i=2, r=r_{m(k)}$ and ${\mathcal U}={\mathcal U}_{m(k)}$), for each $k\geq 1$ there is a closed annulus ${\mathcal W}_{k}\supset {\mathcal U}_{m(k)}$ together with a diffeomorphism $\varphi_{k}:\AR[2/3,6]\rightarrow {\mathcal W}_{k}$ satisfying (P1)-(P3). Because of this, the sequence $({\mathcal W}_{k},g_{r_{m(k)}})$ converges in $C^{2}$ to $(\AR[2/3,6],g_{\mathbb{R}^{3}})$ and $d_{r_{m(k)}}\circ \varphi_{k}$ converges in $C^{0}$ to the distance function to the origin. Then, for $k=k_{0}$ sufficiently big we have
\ben
\varphi_{k_{0}}\big(\AR(4/3,5/3)\big)\subset {\mathcal A}_{r_{m(k_{0})}}(1,4)
\een
and
\ben
{\rm Vol}_{r_{m(k_{0})}}\bigg(\varphi_{k_{0}}\big(\AR(4/3,5/3)\big)\bigg)\geq \frac{{\rm Vol}_{g_{\mathbb{R}^{3}}}\big(\AR(4/3,5/3)\big)}{2}=V_{0}.
\een 
The proposition then follows by defining $\hat{r}_{0}=r_{m(k_{0})}$ and $\hat{\mathcal U}_{0}=\varphi_{k_{0}}\big(\AR(4/3,5/3)\big)$.\end{proof}

We are ready to give the proof of Proposition \ref{AUXWAF}.

\begin{proof}[\bf Proof of Proposition \ref{AUXWAF}.]  
We are going to use repeatedly Lemma \ref{WAFCL} and every time we use it we do with $V=V_{0}$, $\varepsilon=\varepsilon_{0}$, $a=1$, $b=4$ and $i=2$. The reader must keep that in mind because it will not be reminded every time.

First, the conditions (\ref{CONDFI}) in Lemma \ref{WAFCL} are automatically satisfied when we make $r=\hat{r}_{0}$ and ${\mathcal U}=\hat{\mathcal U}_{0}$, where $\hat{r}_{0}$ and $\hat{\mathcal U}_{0}$ are given by Proposition \ref{AUXX}. 
Lemma \ref{WAFCL} then tells that there is ${\mathcal W}$ and $\varphi:\AR[2/3,6]\rightarrow {\mathcal W}$ satisfying (P1)-(P3). 
Make  $\varphi_{0}:=\varphi$ and define ${\mathscr A}_{0}:={\mathcal W}$. 

Second, let $r=\hat{r}_{1}:=2\hat{r}_{0}$ and ${\mathcal U}=\varphi_{0}(\AR(8/3,10/3))$. Then deduce from the definition of $V_{0}$ and $\varepsilon_{0}$ that the conditions (\ref{CONDFI}) in Lemma \ref{WAFCL} are again satisfied (see the comments before the Prop. \ref{AUXX}). Lemma \ref{WAFCL} then tells that there is ${\mathcal W}$ and $\varphi:\AR[2/3,6]\rightarrow {\mathcal W}$ satisfying (P1)-(P3). Make $\varphi_{1}:=\varphi$ and define ${\mathscr A}_{1}:={\mathcal W}$. 

Third, make $r=\hat{r}_{2}=2\hat{r}_{1}=2^{2}\hat{r}_{0}$ and ${\mathcal U}=\varphi_{1}(\AR(8/3,10/3))$. Then deduce from the definition of $V_{0}$ and $\varepsilon_{0}$ that the conditions (\ref{CONDFI}) are again satisfied. 
Lemma \ref{WAFCL} then tells that there is ${\mathcal W}$ and $\varphi:\AR[2/3,6]\rightarrow {\mathcal W}$ satisfying (P1)-(P3). Make $\varphi_{2}:=\varphi$ and define ${\mathscr A}_{2}:={\mathcal W}$.

This procedure can be continued indefinitely obtaining in this way a sequence of closed annuli ${\mathscr A}_{j}$, $j\geq 0$. 
Each ${\mathscr A}_{j}$ has of course two boundary components diffeomorphic to a two-sphere. Denote by $\partial^{+} {\mathscr A}_{j}$ the closest to $\partial E$ and by $\partial^{+} {\mathscr A}_{j}$ the farthest. 
With this notation we have
\be\label{SOI}
{\mathscr A}_{j}\subset {\mathcal A}_{r_{j}}(1/2,8),\quad \partial^{-}{\mathscr A}_{j}\subset {\mathcal A}_{r_{j}}(1/2,1)\quad \text{and}\quad \partial^{+}{\mathscr A}_{j}\subset {\mathcal A}_{r_{j}}(4,8).
\ee
Observing that ${\mathcal A}_{r_{j}}(2a_{1},2a_{2})={\mathcal A}_{r_{j+1}}(a_{1},a_{2})$ for any $a_{2}>a_{1}$ then we have 
\be\label{SOII}
{\mathscr A}_{j+1}\subset {\mathcal A}_{r_{j}}(1,16),\quad \partial^{-}{\mathscr A}_{j+1}\subset {\mathcal A}_{r_{j}}(1,2)\quad \text{and}\quad \partial^{+}{\mathscr A}_{j+1}\subset {\mathcal A}_{r_{j}}(8,16).
\ee
Thus (\ref{SOI}) and (\ref{SOII}) imply 
\be\label{SOIII}
\partial^{-}{\mathscr A}_{j}\cap {\mathscr A}_{j+1}=\emptyset\quad \text{and}\quad \partial^{+}{\mathscr A}_{j+1}\cap {\mathscr A}_{j}=\emptyset
\ee
and as $\varphi_{j}(\AR(8/3,10/3)$ is shared by ${\mathscr A}_{j}$ and ${\mathscr A}_{j+1}$ then we also have ${\mathscr A}_{j}\cap {\mathscr A}_{j+1}\neq \emptyset$. By (\ref{SOIII}) non of the annuli ${\mathscr A}_{j}$ and ${\mathscr A}_{j+1}$ can be contained inside the other and we must have
\ben
\partial^{+}{\mathscr A}_{j}\subset {\mathscr A}_{j+1}\quad \text{and}\quad \partial^{-}{\mathscr A}_{j+1}\subset {\mathscr A}_{j}.
\een
This and (\ref{SOI}) show (Q1). We explain now (Q2). First it is straightforward that that every finite union $\cup_{j=J_{1}}^{j=J_{2}} {\mathscr A}_{j}$ is also a closed annulus and that the infinite union  $\cup_{j=1}^{j=\infty} {\mathscr A}_{j}$ is diffeomoprhic to $\mathbb{R}^{3}$ minus an open ball. 
Then we observe that $\cup_{j=1}^{j=\infty} {\mathscr A}_{j}$ is complete. This is because every Cauchy sequence on it must be also Cauchy in $E$ and therefore uniformly bounded. Thus the sequence must be inside a finite union $\cup_{j=J_{1}}^{j=J_{2}} {\mathscr A}_{j}$ which is complete. Hence the sequence must converge to a point in $\cup_{j=1}^{j=\infty} {\mathscr A}_{j}$. As $E$ is complete and also diffeomorphic to $\mathbb{R}^{3}$ minus a ball then it is standard that $\cup_{j=1}^{j=\infty} {\mathscr A}_{j}$ must cover $E$ up to a set of compact closure. For completeness we indicate a proof of this fact in the auxiliary Proposition \ref{AUXPROP} below.
\end{proof}

\begin{Proposition}\label{AUXPROP} Let $(M,g)$ be a complete Riemannian manifold where $M$ is diffeomorphic to $\mathbb{R}^{3}\setminus B_{\mathbb{R}^{3}}(o,1)$. Suppose that $\Omega\subset M$ is also diffeomoprhic to $\mathbb{R}^{3}\setminus B_{\mathbb{R}^{3}}(o,1)$ and that $(\Omega,g)$ is complete. Then $\overline{M\setminus \Omega}$ is compact in $M$.
\end{Proposition} 

\begin{proof}[\bf Proof.] First note three simple properties derived from the completeness of $(M,g)$ and $(\Omega,g)$ [\footnote{These are straightforward and are left to the reader.}]: (i) Every $C^{1}$-curve $\gamma$ in $M$ starting at a point in $\Omega$ and ending at a point in $M\setminus \Omega$ must cut $\partial \Omega$ at some point, (ii) For any sequence $p_{m}\in \Omega$ such that $\dist_{(\Omega,g)}(p_{m},\partial \Omega)\rightarrow \infty$ then also $\dist_{(M,g)}(p_{m},\partial M)\rightarrow \infty$, and (iii) For any two sequences  $q_{m}$ and $p_{m}$ such that $\dist_{(M,g)}(q_{m},\partial M)\rightarrow \infty$ and $\dist_{(M,g)}(p_{m},\partial M)\rightarrow \infty$ there is a sequence of curves $\gamma_{m}$ joining $p_{m}$ and $q_{m}$ for every $m$ such that $\lim_{m} \dist_{(M,g)}(\{\gamma_{m}\},\partial M)=\infty$ (use that $M\sim (\mathbb{R}^{3}\setminus B_{\mathbb{R}^{3}}(o,1)$)).

If $\overline{M\setminus \Omega}$ is not compact then there is a sequence $q_{m}\in \big(M\setminus \Omega\big)$ with $\dist_{(M,g)}(q_{m},\partial M)\rightarrow \infty$. On the other hand let $p_{m}\in \Omega$ be a sequence such that $\dist_{(\Omega,g)}(p_{m},\partial \Omega)\rightarrow \infty$ and therefore by (ii) with $\dist_{(M,g)}(p_{m},\partial M)\rightarrow \infty$.
By (iii) one can consider curves $\gamma_{m}$ joining $p_{m}$ to $q_{m}$ for which
$\lim_{m} \dist_{(M,g)}(\gamma_{m},\partial M)=\infty$.
But by (i) every $\gamma_{m}$ must cut $\partial \Omega$ and therefore we must have $\dist_{(M,g)}(\gamma,\partial M)\leq \max\{\dist_{(M,g)}(p,\partial M),p\in \partial \Omega\}<\infty$ for every $m$. We reach thus a contradiction.
\end{proof}

With the help of Lemma \ref{WAFCL} and Proposition \ref{AUXWAF} we can now prove the main result of Part II.

\begin{Theorem}\label{MAINTII} Let $E$ be a strongly stationary end having cubic volume growth. Then, $E$ is WAF.
\end{Theorem} 

\begin{proof}[\bf Proof.] Assume that integers $i\geq 2$ and $l\geq 1$ are given, as well as a divergent sequence $r_{m}\rightarrow \infty$ ($m\geq 1$). According to the Definition \ref{WAFD}, to show weakly asymptotic flatness we need to show the existence of open and connected regions $\Omega_{m}$ for which (W1)-(W3) hold. To define the $\Omega_{m}$ we will rely in the following claim: {\it For any $k\geq 4$ there is $m_{k}>0$ such that for any $m\geq m_{k}$ there is $\tilde{\Omega}_{k,m}$ and $\tilde{\varphi}_{k,m}:\AR(1/2-1/k,2^{l}+1/k)\rightarrow \tilde{\Omega}_{k,m}$ satisfying
\begin{enumerate}[labelindent=\parindent,leftmargin=*,label={\rm (W\arabic*')}]
\item ${\mathcal A}_{r_{m}}(1/2,2^{l})\subset \tilde{\Omega}_{k,m}$ for every $m\geq m_{k}$, and
\item $\tilde{\varphi}_{k,m}^{*}(g_{r_{m}})$ is $1/k$-close in the $C^{i}$-norm to the Euclidean metric, and
\item $d_{r_{m}}\circ \tilde{\varphi}_{k,m}$ is $1/k$-close in the $C^{0}$-norm to the distance function to the origin.
\end{enumerate}
}
\n Once the claim is proved the construction of the regions $\Omega_{m}$ is as follows (assume, redefining $m_{k}$ if necessary, that $m_{k+1}>m_{k}$, for all $k\geq 1$). For $m$ between $m_{1}$ and $m_{2}-1$ let $\Omega_{m}=\tilde{\Omega}_{1,m}$; for $m$ between $m_{2}$ and $m_{3}-1$ let $\Omega_{m}=\tilde{\Omega}_{2,m}$; for $m$ between $m_{3}$ and $m_{4}-1$ let $\Omega_{m}=\tilde{\Omega}_{3,m}$ and so on (for $m$ between $1$ and $m_{1}-1$ define $\Omega_{m}$ as any annulus containing ${\mathcal A}_{r_{m}}(1/2,2^l)$). With this definition of $\Omega_{m}$, (W1')-(W3') imply (W1)-(W3) directly.

We prove now the claim. Because $r_{m}\rightarrow \infty$ and because of (Q2) we can assume without loss of generality that ${\mathcal A}_{r_{m}}(1/2,2^{l})\subset \cup_{j=1}^{j=\infty} {\mathscr A}_{j}$. Then, for every $m$ define $j_{m}$ such that 
\be\label{SOBERANA}
8\, (2^{j_{m}} \hat{r}_{0})< \frac{r_{m}}{2} \leq 8\, (2^{j_{m}+1} \hat{r}_{0}).
\ee
Recalling that $\hat{r}_{j}=2^{j}\hat{r}_{0}$, this says that $8 \hat{r}_{j_{m}}<r_{m}/2$. But then, from (Q1) we obtain that 
\ben
\bigg(\cup_{j=1}^{j=j_{m}} {\mathscr A}_{j}\bigg)\cap {\mathcal A}_{r_{m}}(\frac{1}{2},2^{l})=\emptyset.
\een
Also from (\ref{SOBERANA}) we get $2^{l}r_{m}<2^{j_{m}+l+6}\hat{r}_{0}$, which implies $2^{l}r_{m}<\hat{r}_{j_{m}+l+5}/2$. But then from (Q1) we obtain that
\ben
\bigg(\cup_{j=j_{m}+l+5}^{j=\infty} {\mathscr A}_{j}\bigg)\cap {\mathcal A}_{r_{m}}(\frac{1}{2},2^{l})=\emptyset.
\een 
We conclude that 
\be\label{CONDIFIII}
{\mathcal A}_{r_{m}}(\frac{1}{2},2^{l})\subset \bigg(\cup_{j=j_{m}+1}^{j=j_{m}+l+4} {\mathscr A}_{j}\bigg)\subset {\mathcal A}(\frac{\hat{r}_{j_{m}+1}}{2},8 \hat{r}_{j_{m}+l+4}) \subset {\mathcal A}_{r_{m}}(\frac{1}{32},2^{l+3})
\ee
where the second inclusion is because of (Q1) and the third is because $r_{m}/32\leq \hat{r}_{j_{m}+1}/2$ and $8\hat{r}_{j_{m}+l+4}\leq 2^{l+3}r_{m}$ which are deduced from (\ref{SOBERANA}).

As the end $E$ has cubic volume growth, then $\lim_{r_{m}\rightarrow \infty} {\rm Vol}_{r_{m}}\big({\mathcal A}_{r_{m}}(1/2,2^{l})\big)$ is positive and we can assume that ${\rm Vol}_{r_{m}}\big({\mathcal A}_{r_{m}}(1/2,2^{l})\big)\geq \mu>0$ for all $m$ (see a similar argument in the proof of Prop. \ref{AUXX}). 
Now, for every integer $k\geq 4$, let $r_{0}(V_{0},\varepsilon_{0},a,b,i)$ be the $r_{0}$ provided by Lemma \ref{WAFCL} with the following values of $V_{0},\varepsilon_{0},a$, $b$ and $i$:  $V_{0}=\mu$, $\varepsilon_{0}=1/k$, $a=1/32$, $b=2^{l+3}$ and $i=2$. As we only let $k$ to vary we can denote $r_{0}(V_{0},\varepsilon_{0},a,b,i)=r_{0}(k)$. Then for every $k\geq 4$ define $m_{k}$ such that for every $m\geq m_{k}$ we have $r_{m}\geq r_{0}(k)$.   
Then for every $m\geq m_{k}$ the region ${\mathcal U}:= {\rm Int}\big(\cup_{j=j_{m}+1}^{j=j_{m}+l+4} {\mathscr A}_{j}\big)$ is open by definition, is connected because of (Q2) and verifies (\ref{CONDFI}) by (\ref{CONDIFIII}). We can then apply Lemma \ref{WAFCL} 
and conclude that there is ${\mathcal W}$ and $\varphi: \AR(2a/3,3b/2)\rightarrow {\mathcal W}$ satisfying (P2)-(P3) with, as we are assuming, $\varepsilon_{0}=1/k$.   
With all this at hand define $\Omega_{k,m}=\varphi(\AR(1/2-1/k,2^{l}+1/k))$ for any $m\geq m_{k}$. With this definition (W2') and (W3') follow directly from (P2) and (P3), and (W1') is easily seen to follow from (P3). 
\end{proof}

\bibliographystyle{plain}
\bibliography{Master.bib}

\end{document}